\documentclass[11pt]{article}
\usepackage[tmargin=1in,lmargin=1in,rmargin=1in,bmargin=1.2in]{geometry}
\usepackage{xcolor}
\usepackage{amssymb,amsmath,amsthm}
\usepackage{xspace}
\usepackage{thmtools}
\usepackage{graphicx}
\usepackage{tcolorbox}
\tcbuselibrary{skins,breakable}
\tcbset{enhanced jigsaw}
\usepackage{nicefrac}

\usepackage{nameref}
\usepackage[linktocpage=true,
	pagebackref=true,colorlinks,
	linkcolor=Red,citecolor=ForestGreen,
	bookmarks,bookmarksopen,bookmarksnumbered]
	{hyperref}
\usepackage[noabbrev,nameinlink]{cleveref}
\crefname{property}{property}{Property}
\creflabelformat{property}{(#1)#2#3}

\crefname{equation}{eq}{Eq}
\creflabelformat{equation}{(#1)#2#3}

\crefname{algocf}{Algorithm}{Algorithm}
\creflabelformat{algocf}{(#1)#2#3}

\newtheorem{theorem}{Theorem}

\newtheorem{lemma}{Lemma}[section]
\newtheorem{proposition}{Proposition}[section]
\newtheorem{claim}[lemma]{Claim}

\usepackage{bbm}
\usepackage{enumitem}

\theoremstyle{definition}

\newtheorem{remark}[lemma]{Remark}

\theoremstyle{definition}

\newtheorem{mdalg}{Algorithm}
\newenvironment{Algorithm}{\begin{tbox}\begin{mdalg}}{\end{mdalg}\end{tbox}}

\newtheorem{mddist}{Distribution}

\usepackage[normalem]{ulem}
\usepackage[compact]{titlesec}

\definecolor{ForestGreen}{rgb}{0.1333,0.5451,0.1333}
\definecolor{Red}{rgb}{0.9,0,0}
\definecolor{DarkRed}{rgb}{0.5,0.1,0.1}
\definecolor{RURed}{rgb}{0.95,0.1,0.1}
\definecolor{DarkBlue}{rgb}{0.1,0.1,0.5}

\newcommand{\Ot}{\ensuremath{\widetilde{O}}}

\newcommand{\eps}{\ensuremath{\varepsilon}}

\newcommand{\bracket}[1]{\left[#1\right]}
\newcommand{\paren}[1]{\ensuremath{\left(#1\right)}\xspace}
\newcommand{\card}[1]{\left\vert{#1}\right\vert}

\newcommand{\prob}[1]{\Pr\paren{#1}}
\newcommand{\expect}[1]{\Exp\bracket{#1}}

\newcommand{\set}[1]{\ensuremath{\left\{ #1 \right\}}}
\newcommand{\poly}{\mbox{\rm poly}}
\newcommand{\polylog}[1]{\textnormal{polylog}\,{#1}\xspace}

\newcommand{\ALG}{\ensuremath{\mbox{\sc alg}}\xspace}

\DeclareMathOperator*{\Exp}{\ensuremath{{\mathbb{E}}}}
\DeclareMathOperator*{\Prob}{\ensuremath{\textnormal{Pr}}}
\renewcommand{\Pr}{\Prob}

\newenvironment{tbox}{\begin{tcolorbox}[
		enlarge top by=5pt,
		enlarge bottom by=5pt,
		 breakable,
		 boxsep=2pt,
                  left=5pt,
                  right=7pt,
                  top=10pt,
                  arc=0pt,
                  boxrule=1pt,toprule=1pt,
                  colback=white
                  ]%%
	}
{\end{tcolorbox}}

\newcommand{\II}{\ensuremath{\mathbb{I}}}

\newcommand{\mireal}[1][]{
  \ifx\relax#1\relax%
    \II(\mione \,; \mitwo)%
  \else%
    \II(\mione \,; \mitwo\mid #1)%
  \fi
}

\newcommand{\dist}{\ensuremath{\mathcal{D}}}

\renewcommand{\leq}{\leqslant}

\renewcommand{\geq}{\geqslant}

\newcommand{\myqed}[1]{\let\qed\relax \hspace*{\fill} #1 \ensuremath{\square}}

\newcommand{\orc}{\ensuremath{\texttt{ORC}_{G, I^*}}\xspace}

\newcommand{\nbri}[1]{\ensuremath{\deg_{I^*}(#1)}}
\newcommand{\nbrit}[1]{\ensuremath{\widetilde{\deg}_{I^*}(#1)}}

\title{Learning-augmented Maximum Independent Set}

\author{%
\begin{tabular}{cc}
\begin{tabular}[t]{c}
Vladimir Braverman\footnote{\href{mailto:vb21@rice.edu}{vb21@rice.edu}. Supported partially by the Naval Research (ONR) grant N00014-23-1-2737 and NSF-CNS 2333887 award.}\\ Rice University and Google Research  \\
\end{tabular} &
\begin{tabular}[t]{c}
Prathamesh Dharangutte \footnote{\href{mailto:prathamesh.d@rutgers.edu}{prathamesh.d@rutgers.edu}. Supported by NSF through IIS-2229876 and CCF-2118953.
} \\ Rutgers University
\end{tabular} \\ %\addlinespace[4ex] % create some vertical separation
\rule{0pt}{3ex}
\begin{tabular}[t]{c}
Vihan Shah\footnote{\href{mailto:vihan.shah@uwaterloo.ca}{vihan.shah@uwaterloo.ca}. Supported in part by Sepehr Assadi's Sloan Research Fellowship and NSERC Discovery 
Grant.} \\ University of Waterloo \\
\end{tabular} &
\begin{tabular}[t]{c}
Chen Wang\footnote{\href{mailto:cw200@rice.edu}{cw200@rice.edu}.} \\ Rice University and Texas A\&M University\\
\end{tabular}
\end{tabular}
}

\date{}

\begin{document}

\maketitle

\begin{abstract}
    We study the Maximum Independent Set (MIS) problem on general graphs within the framework of learning-augmented algorithms.
    The MIS problem is known to be NP-hard and is also NP-hard to approximate to within a factor of $n^{1-\delta}$ for any $\delta>0$. 
    We show that we can break this barrier in the presence of an oracle obtained through predictions from a machine learning model that answers vertex membership queries for a fixed MIS with probability $1/2+\eps$.
    In the first setting we consider, the oracle can be queried once per vertex to know if a vertex belongs to a fixed MIS, and the oracle returns the correct answer with probability $1/2 + \eps$. 
    Under this setting, we show an algorithm that obtains an $\Ot(\sqrt{\Delta}/\eps)$\footnote{Throughout we use $\Ot(\cdot)$ to hide $\polylog(n)$ factors.}-approximation in $O(m)$ time where $\Delta$ is the maximum degree of the graph.
    In the second setting, we allow multiple queries to the oracle for a vertex, each of which is correct with probability $1/2 + \varepsilon$. For this setting, we show an $O(1)$-approximation algorithm using $O(n/\eps^2)$ \textit{total} queries and $\Ot(m)$ runtime. 
    % These results show significantly better approximation since for the classical MIS problem there is no $n^{1-\delta}$ approximation algorithm for any $\delta > 0$ unless P=NP.

    % The MIS problem is known to be NP-hard and notoriously difficult to approximate. It is NP-hard to approximate the size of the MIS to within a factor of $n^{1-\delta}$ for any $\delta>0$. 
    % We show that in the presence of an oracle that can correctly predict whether a vertex belongs to a fixed MIS with probability $1/2+\eps$, we can get an $\Ot(\sqrt{\Delta}/\eps)$-approximation in $O(m)$ time.
    % We also show that if we can query the oracle multiple times for a single vertex with fresh randomness then we can get a constant approximation to the MIS using $O(n/\eps^2)$ total queries in $\Ot(m)$ time. Notice that if $\eps = \Theta(1)$ then we can get a constant approximation in the same number of queries as in the model with single queries where we could only get a $\Ot(\sqrt{n})$ approximation.
\end{abstract}

\section{Introduction}
\label{sec:intro}
We consider learning-augmented \emph{maximum independent set} (MIS) in this paper. Given a (unweighted, undirected) graph $G=(V,E)$, an independent set is a set of vertices $I \subseteq V$, such that for any $u,v\in I$, $(u,v)\not\in E$, i.e., there is \emph{no} edge between $u$ and $v$. The maximum independent set problem asks to find the independent set with the largest size in $G$.

% \vihan{the following part really seems like chatgpt lol}
% \chen{Yeah, I used chatgpt to polish the writing :)}
% \prat{simplified}
Finding the maximum independent set is one of the classical NP-hard problems \cite{karp2010reducibility}. Furthermore, the seminal work of \cite{Hastad96,zuckerman2006linear} demonstrates the NP-hardness of approximating the size of the MIS to within a factor of $n^{1-\delta}$ for any $\delta>0$. 
In contrast, outputting any single vertex gives an $n$-approximation trivially.
\cite{boppana1992approximating} gave a non-trivial $O(n/\log^2 n)$-approximation to MIS and this was later improved by \cite{feige2004approximating}.
These results indicate that the problem is quite hard in its general form and thus, many research efforts have been devoted to approximation algorithms in special settings, e.g., planar graphs \cite{AlekseevLMM08,MagenM09}, rectangle-intersection graphs \cite{ChalermsookC09,ChuzhoyE16,GalvezKMMPW22}, and exponential-time algorithms \cite{robson2001finding,FominGK06,XiaoN13,BourgeoisEPR12}.

% This problem boasts a rich history of exploration, standing as one of the classical NP-hard problems. Furthermore, it proves resistant to approximation, as evidenced by its well-documented hardness in the literature. Notably, the seminal work of \cite{Hastad96} demonstrates the NP-hardness of approximating the size of MIS within a factor of $n^{1-\eps}$ for any $\eps>0$. 

On the other hand,  heuristic algorithms, despite their bad worst-case guarantees, often exhibit commendable performance on real-world graphs \cite{andrade2012fast,DahlumLS0SW16,WalterosB20}. 
For instance, the greedy algorithm only offers an approximation guarantee of $O(\Delta)$, where $\Delta$ is the \emph{maximum degree} of $G$. 
However, it frequently yields satisfactory empirical results. 
%Furthermore, more recently, heuristic algorithms with machine learning-based techniques have also been considered (see, e.g.,~\cite{AhnSS20,PontoizeauSYC21,BruscaQS0C23}). These methods have demonstrated even better performances and efficiency. 
The gap between the worst-case hardness and practical efficiency motivates us to study the MIS problem through the lens of beyond worst-case analysis \cite{boppana1987eigenvalues,roughgarden2021beyond}. In particular, under the modern context, we ask the question of finding the maximum independent set with \emph{learning-augmented oracles}.

\paragraph{Learning-augmented algorithms.} Learning-augmented algorithms, also known as algorithms with predictions, have attracted considerable attention in recent years (see, e.g.~\cite{PurohitSK18,HsuIKV19,LavastidaM0X21,ChristiansonHW22,SW23Learning,BGTZ23Online,BanerjeeC0L23,ACNSV23Improved,henzinger2023complexity,brand2024dynamic,sadek2024algorithms}, and references therein). This paradigm of beyond worst-case analysis has been successful in surmounting classical thresholds and bridging the gap between the worst-case hardness and practical efficiency (see, e.g., \cite{mitzenmacher2022algorithms}, for an excellent summary).
Typically, in learning-augmented algorithms, we assume the access to an oracle that gives part of the ``right answer'' to the problem, and fails with some small but non-negligible probability. Conceptually, these algorithms aim to take advantage of modern machine learning models, which are fairly accurate on predictive tasks yet make random mistakes in an inconsistent fashion. Learning-augmented algorithms provide a great way to analyze algorithms beyond the worst case, 
and these algorithms usually have immediate implications in practice (see the empirical results in, e.g., \cite{Chledowski0SZ21,HsuIKV19,ErgunFSWZ22,SW23Learning,ACNSV23Improved}). Inspired by the recent work in utilizing machine learning-based techniques for the maximum independent set ~\cite{AhnSS20,PontoizeauSYC21,BruscaQS0C23}, we consider the MIS problem through the lens of learning-augmented algorithms.
% \chen{emphasize the fact we need new models since models for learning-augmented algorithms are ad-hoc.}

The advantage of the learning-augmented algorithms has inspired a flurry of work that studies \emph{graph problems} within this framework ~\cite{bamas2020primal,feijen2021using,chen2022faster,chen2022triangle, hu2023connectivity, antoniadis2023online, lattanzi2023speeding, azar2024discrete,COGLP24maxcut,ghoshal2024constraint}. In a very recent work, \cite{COGLP24maxcut} considered the Max-cut problem, in which the oracle model is closely related to our setting for the MIS problem. Under the Unique Game Conjecture (UCG), it is known that getting anything better than $\alpha \approx 0.878$ approximation for max-cut is NP-hard (\cite{khot2007optimal}). In contrast, \cite{COGLP24maxcut} showed that with a learning-augmented oracle, we could achieve better approximation than the $\alpha$ threshold in polynomial time. In another closely related work, \cite{ghoshal2024constraint} studied the more general constraint satisfaction problems (CSPs) trough the lens of the learning-augmented algorithms. There, they obtain results for both the Max-cut and the Max 2-Lin problem. Although \cite{COGLP24maxcut,ghoshal2024constraint} defines more general learning-augmented oracles, they, unfortunately, fall short of capturing the MIS-type of CSP problems, and their results do not have direct implications on the MIS problem. % \textcolor{red}{Chen: We should expand the discussion according to the reviewers' comments.}

% \cite{COGLP24maxcut,ghoshal2024constraint} study the Max-Cut problem which is also classically NP-hard, in the presence of predictions. They give a polynomial time algorithm to approximate the Max-Cut beyond the factor of $\alpha \approx 0.878$ which is the Max-Cut threshold (optimal under the Unique Games Conjecture \cite{khot2007optimal}).
% \vihan{Chen will make another pass and discuss more about the CSP stuff}

\begin{comment}
Since the maximum independent set problem could also be viewed as a constraint satisfaction problem, one may naturally wonder whether approximation algorithms are directly implied by the results of \cite{COGLP24maxcut,ghoshal2024constraint}. Unfortunately, the application of \cite{COGLP24maxcut,ghoshal2024constraint} to the MIS problem is far from straightfoward. To elaborate, most of the learning-augmented algorithms use oracles that tailor to specific applications. As such, oracles for problems like Max-cut (e.g., in \cite{COGLP24maxcut}) cannot be directly used in the MIS problem. The results of \cite{ghoshal2024constraint} are indeed more general; nevertheless, their algorithm (with the predictions) remains restricted to the settings they mentioned, and it is challenging to generalize to the MIS problem \vihan{Is it more general?}. 
\end{comment}

From the above discussion, we can see that $a)$ studying the maximum independent set problem in the framework of learning-augmented algorithms has great potential; and $b)$ to this end, the existing models and algorithms are not yet sufficient.
In light of this, we ask the following question:

\begin{center}
\emph{Under the framework of learning-augmented algorithms, what efficient algorithms can we get for the maximum independent set problem?}
\end{center}

\subsection{Our models and contributions}
In what follows, we will define the learning-augmented oracle model we consider and present our main results.

\paragraph{Our oracle model.}
We consider the following natural learning-augmented oracle: for a fixed maximum independent set $I^*$, the oracle answers whether a vertex $v\in I^*$ correctly with probability $1/2+\eps$, and incorrectly with probability $1/2-\eps$. In addition, the randomness is \emph{independent} across the vertices. 
We denote by $\orc(v)$ the answer the oracle gives when queried for vertex $v$.

%\paragraph{Settings and results.}
We study approximation algorithms for MIS with the learning-augmented oracle in two settings: the \emph{persistent noise} setting and the \emph{non-persistent noise} setting. We discuss the settings and the results, respectively, as follows.
\begin{itemize}
\item {The \textbf{persistent noise} setting.} In the persistent noise setting, the randomness of $\orc$ is drawn exactly \emph{once}. 
Therefore, the answer for a vertex will remain the same no matter how many times you query the oracle.
Another way to think about this is that the oracle can be queried at most once for a vertex.
This setting is the most standard in the learning-augmented literature, and graph problems are often studied under persistent noise (see, e.g.~\cite{feijen2021using,chen2022faster,chen2022triangle,XiaH22,hu2023connectivity, antoniadis2023online,COGLP24maxcut,ghoshal2024constraint} and references therein).
Our main result in this setting is a randomized algorithm that with high probability\footnote{As standard in the literature, we use ``with high probability'' to denote a success probability of $1-{1}/{\poly{(n)}}$.} achieves an $\Ot(\sqrt{\Delta})$ (multiplicative) approximation to the MIS in $O(m)$ time (\Cref{thm:persist-noise}).
% \vihan{Do we want to call this the MAB setting? I think something like single query vs multiple queries will be better} 
% \chen{That makes sense. The reason for me to call it MABs setting now is because we used techniques heavily from the MABs literature, and the fusion of MABs and graph MIS techniques is the ``highlight'' of our algorithm. Let me think about it.}
\item {The \textbf{non-persistent noise} settings.} In this setting, for each vertex $v$, we allow $\orc(v)$ to use \emph{fresh randomness} for different queries. If we are allowed to make $O(n \log{n})$ queries to the oracle in total, then we can trivially recover the entire set $I^*$ with high probability by querying each vertex $O(\log n)$ times. The interesting case is when we are allowed to make only $O(n)$ queries, i.e., a number that is \emph{asymptotically the same} as the persistent noise setting. % The question then is can we get a better approximation? 
Although the non-persistent noise setting is less frequently studied in the learning-augmented algorithm literature, it has recently sparked considerable interest in various problems \cite{GulloMT23,GuptaSM24,KMBC24query}.
In \Cref{app:mabs-simple-alg}, we show that it is easy to get an $O(\log{n})$-approximation with $O(n)$ queries. Our main result considerably improved on the approximation factor: we show that we can indeed obtain an $O(1)$ approximation with $O(n)$ queries and $\Ot(m)$ runtime (\Cref{thm:MABs-setting}).
\end{itemize}

Our results in the persistent noise setting hold assuming full independence, but it can be easily extended to the setting where oracle queries are assumed to use $k$-wise independent hash function for $k=O(\log n)$. Extending it to the pair-wise independent case is challenging as the failure probabilities in the concentration bounds are not enough for the application of a union bound.

% \chen{Discuss: 1). why flip the answer as opposed to grant the adversary power to answer arbitrarily; 2). why not using random answers; and 3). why not using with probability at most $(1-p)$ (you can, but it has to be the same across vertices). }

\subsection{Technical overview}
% \chen{Maybe move the technical overview here.} \prat{\checkmark}
The biggest challenge in leveraging the oracle information is distinguishing the case where $\orc(v)$ is indeed correct. In what follows, we give a high-level overview of our techniques describing how we can use the neighborhood information for this purpose. For the simplicity of the discussion, we always assume $\eps=\Theta(1)$ in the technical overview.

\paragraph{Persistent noise setting:} %We want to utilize the signal from oracle to get a good approximation to the MIS. 
A natural approach in this setting would be to figure out the conditions in which a ``yes'' signal for a vertex $v$ from the oracle implies $v\in I^*$, by aggregating signals from $N(v)$. However, such an idea is hard in the following sense. For a vertex $v$ whose oracle query $\orc(v)=1$, if there are many $u\in N(v)$ such that $\orc(u)=1$, we can determine that $v$ should \emph{not} be in the MIS. However, the converse is not true: if a vertex $v$ is \emph{not} in $I^*$, it does \emph{not} necessarily have many neighbors in $I^*$. As a result, simply aggregating neighborhood information might not be enough to determine the membership of a vertex in the MIS.

%Therefore, it is unclear how to use `correct' the answers from the $\orc$ oracle simply by aggregating neighborhood information.  

The key idea here is, instead of looking at the oracle answer for vertex $v$ ($\orc(v)$), we look at what the oracle says for the \emph{neighborhood} of the vertex $v$. This turns out to be a good enough signal to eliminate vertices that have \emph{many} edges to the MIS $I^*$. 
Specifically, we can show that if $v$ has $\widetilde{\Omega}(\sqrt{\deg(v)})$ edges to $I^*$, then the oracle queries for $N(v)$ contain enough information to identify such a vertex $v$.
%we can catch it from the answers of $\orc(u)$ for $u\in N(v)$.
Upon removal of such vertices, the remaining vertices have a small degree ($\Ot(\sqrt{\Delta})$) to $I^*$, and a greedy independent set on the residual vertices gives a good approximation.

\paragraph{Non-persistent noise setting: } 
Our algorithm for this setting is a bit more nuanced as we aim to minimize the query complexity to the oracle while aiming to achieve a good approximation.
The starting point of our algorithm is from the viewpoint of the classical \emph{pure exploration} algorithms in \emph{multi-armed bandits (MABs)}. If we ignore the nature of MIS in our oracle, we can reduce to the following MABs problem: given $n$ arms with mean rewards as either $\frac{2}{3}$ or $\frac{1}{3}$, find \emph{all} the arms with mean reward $\frac{2}{3}$ with $O(n)$ arm pulls. 
It is well-known that one can find a \emph{single} best arm with high constant probability in $O(n)$ queries. The question is, can we solve the problem by resorting to purely MABs algorithms, and simply ignoring the nature of the MIS?

It turns out that the above plan is not generally feasible. In particular, we note that returning the set of \emph{all} arms with the higher reward is very similar to finding the \emph{top-$k$ arms} in the MABs literature (see, e.g.~\cite{KalyanakrishnanS10,KalyanakrishnanTAS12,CaoLTL15,ChenLQ17,SimchowitzJR17}). In general, it would require $\Omega(n\log{k})$ arm pulls to obtain top-$k$ arms with high constant probability (\cite{KalyanakrishnanTAS12,SimchowitzJR17}). In \Cref{app:lb-mabs}, we provide lower bound results, showing that to find even $O(1)$ fraction of the high-reward arms in the instance distribution requires $\omega(n)$ queries. The lower bounds teach us that to obtain the desired query efficiency and approximation guarantee, we have to exploit the structure of the MIS.

To better understand the hardness and the insights of MABs algorithms on our problem, let us look at the elimination-based algorithm as in the classical algorithm of \cite{Even-DarMM02,Even-DarMM06}. The first idea we can try is to adapt the elimination algorithm to our problem. To this end, a natural idea is to perform elimination based on whether the mean empirical reward of an arm is more than $\frac{1}{2}$. More concretely, we maintain a pool $\tilde{I}$ of surviving vertices and use $s_r$ as the number of queries to each vertex in round $r$ with $s_1 = O(1)$. In round $r$, we can query $\orc(v)$ for $s_{r}$ time for each $v\in \tilde{I}$. We then eliminate all vertices $v \in \tilde{I}$ whose number of `yes' answers is less than $s_{r}/2$, and recurse on the new $\tilde{I}$ to round $r+1$, for which we set $s_{r+1} = 1.5 s_r$.

Since the probability for any $v\not\in I^*$ to survive decreases doubly-exponentially with the number of rounds, we can show that \textit{all} vertices $v\not\in I^*$ are eliminated after $O(\log\log{n})$ rounds, and the total sample complexity on the \emph{non-MIS} vertices is at most $O(n)$. Furthermore, the probability of losing any $v\in I^*$ decreases exponentially, we can argue that in the end, $\tilde{I}$ contains at least $\Omega(1)$ fraction of the vertices in $I^*$. Unfortunately, due to this fact, for each vertex $v\in I^*\cap \tilde{I}$, i.e., the vertices in the MIS that survive till the end, we need to pay for $2^{O(\log\log{n})}=\polylog{n}$ on the sample complexity. Therefore, this pure exploration algorithm only works when the size of $I^*$ is upper-bounded by $n/\polylog{n}$, and its worst-case guarantee is only a $\polylog{n}$ approximation.

Note that a $\polylog{n}$ approximation is far from what we want: after all, there is a trivial algorithm that achieves $O(\log{n})$ approximation with $O(n)$ samples (see~\Cref{app:mabs-simple-alg} for details). Nevertheless, the existence of such an algorithm teaches us that the problematic case is when the MIS size is large and, in particular, \emph{comparable} to the size of the non-MIS vertices. As such, a natural idea is to design an algorithm that handles the case when the numbers of the MIS and the non-MIS vertices are comparable, and fuse this algorithm with the elimination-based MABs procedure we discussed above.

The above idea is quite close to the final strategy we adapt, albeit we proceed differently for the roles of the two components. In particular, we use the pure exploration MABs algorithm not to output a set with vertex set $\tilde{I}\subseteq I^*$, but to output a set of vertex set $\tilde{I}$ whose \emph{majority (but not necessarily all) of} vertices are in $I^*$. To this end, we use a more \emph{conservative} elimination strategy than the ones in the line of \cite{Even-DarMM02,Even-DarMM06}: instead of increasing the number of samples by a multiplicative factor, we increase the number of samples in each round by an \emph{additive} factor. In this way, we cannot guarantee that all the ``wrong'' arms are eliminated; however, we can argue that, since the probability for the non-MIS vertices to survive decreases exponentially, we have $i).$ the number of samples used on the vertices in $I^*$ is bounded by $O(n)$ \emph{before} the size of $\tilde{I}\setminus I^*$ reduces to the size of $\tilde{I}\cap I^*$; and $ii).$ the number of vertices in $\tilde{I}\cap I^*$ only decreases by a constant fraction. In this way, we can design an efficient procedure that eliminates the ``surplus'' non-MIS vertices to always create cases when the number of non-MIS vertices is smaller.

The final missing piece is the MIS algorithm that deals with the case when the number of MIS vertices takes the majority of the vertex set. Our algorithm to handle this case is to compute an \emph{approximate vertex cover} of the graph and the remaining vertices will form an approximate independent set. It is a well-known fact that if we compute a \emph{maximal matching} and take \emph{all} their endpoints, it forms a $2$-approximate vertex cover that covers all edges in the graph. Furthermore, since the size of the non-MIS vertices is small, there can be only a \emph{limited number} of vertices $v\in I^*$ that can be counted in the vertex cover. As such, we can simply \emph{remove} these vertices from the graph. The rest of the graph would form an independent set, and since we remove at most a constant fraction of vertices from $I^*$ throughout the two phases, we get an $O(1)$ approximation.

\section{Preliminaries}
\label{sec:prelim}

\paragraph{Notation.}
For a graph $G=(V,E)$, we use $\deg(v)$ and $N(v)$ for each vertex $v \in V$ to denote the degree 
and neighborhood of $v$, respectively. We use $G[U]$ for any set $U$ of vertices to denote the induced subgraph of $G$ on $U$. 

We let $I^*$ denote a fixed maximum independent set of the graph $G$. We let $N_{I^*}(v) = N(v) \cap I^*$ be the set of neighbors of the vertex $v$ in the independent set and let $\nbri{v}:= \card{N_{I^*}(v)}$ be its size.
Furthermore, we let $\widetilde{N_{I^*}}(v)$ be the set of neighbors of the vertex $v$ that are claimed to be in the independent set by the oracle and let $\nbrit{v}$ be its size.

For the purpose of conciseness, we defer the technical preliminaries to \Cref{app:tech-prelim}.
% \vihan{Is this notation okay? Also, we say $I^*$ for the MIS but the subscript in the notation is $I$} \prat{Changed it.}

%\prat{Do we need this in main text?}

\section{An Algorithm in the Persistent Noise Setting}
\label{sec:alg-persist}

%We present our main result for the learning-augmented maximum independent set in the persistent noise setting in this section. The formal statement is as follows.

In this section we present an algorithm for the learning-augmented MIS problem with persistent noise. Formally we prove the following

\begin{theorem}
\label{thm:persist-noise}
There exists a randomized algorithm that given
\begin{enumerate}[label=\alph*)]
\item an input graph $G=(V,E)$ with maximum degree $\Delta$ and 
\item an MIS oracle $\orc$ with \emph{persistent noise} for an unknown maximum independent set $I^*$,
\end{enumerate}
in $O(m)$ time outputs an independent set $I$ such that $\card{I}\geq \frac{\eps}{12} \cdot (\Delta \ln n)^{-0.5} \cdot \card{I^*}$ with high probability.
\end{theorem}

We dedicate the remainder of this section to the proof of \Cref{thm:persist-noise}.
We start with the assumption that $\eps\leq 1/4$ (we can do this for any constant $>0$).
This assumption is needed for technical reasons.
If $\eps>1/4$, then it is easy to simulate an oracle with $\eps=1/4$ by flipping the oracle answer with probability $p=\frac{\eps-1/4}{1/2+\eps}$ ($p\geq 0$ since $\eps>1/4$).
If we do this then the probability that the oracle gives the incorrect answer is $(1/2-\eps) + p \cdot (1/2+\eps) = 1/4$ which is exactly what we wanted.
Note that the final bound we get on the approximation factor now changes by a factor of at most $2$.
This is because when $\eps>1/4$ we are replacing it with an oracle for $\eps=1/4$ and the approximation factor linearly depends on $\eps$.

\paragraph{The algorithm and analysis.} %We now give the formal algorithm and the analysis. The algorithm is as follows.
We now state our algorithm.
% \chen{I moved the algorithm box here -- always give the algorithm before giving the analysis.} \prat{\checkmark}
\begin{Algorithm}
\label{alg:mis-persistent-noise}
%A learning augmented algorithm that computes an $O(\sqrt{\Delta \log{n}}/\eps)$-approximation of the maximum independent set.\\
 An algorithm for MIS in persistent noise setting.\\
\textbf{Input:} A graph $G=(V,E)$ with maximum degree $\Delta$ that contains an unknown maximum independent set $I^*$; an MIS oracle $\orc$ in the persistent noise setting \\
\textbf{Output}: A set of vertices $I$ such that $I$ forms an independent set and $\card{I}\geq \frac{\eps}{3} \cdot (\Delta \ln n)^{-0.5} \cdot \card{I^*}$.\\
\textbf{Parameters:} $s_{v}:= (1/2-\eps)\deg(v) + 6\sqrt{\ln n} \cdot (1/2-\eps) \sqrt{\deg(v)}$ . 
\smallskip
\begin{enumerate}
    \item Calculate $\nbrit{v}$ for all vertices $v\in V$.
    \item Let $L$ be the set of vertices where $\deg(v)\leq 36 \ln n$ for $v\in V$.
    \item Let $S$ be the set of vertices where $\nbrit{v} \leq s_v$ for $v\in V \setminus L$.
    \item Output the greedy MIS $I$ on $G[S \cup L]$.
\end{enumerate}
\end{Algorithm}

%We start our analysis by showing that if $v$ is in the MIS set of $I^*$, the number of ``yes'' answers in the neighborhood of $v$ cannot be too high.

We first show that if $v \in I^*$, the number of ``yes'' answers in $N(v)$ cannot be too high.
\begin{claim}\label{clm:vert-in-I}
    If $v\in I^* \setminus L$ then with high probability, $\nbrit{v} \leq (1/2-\eps)\deg(v) + 6\sqrt{\ln n} \cdot (1/2-\eps) \sqrt{\deg(v)}$.
\end{claim}
\begin{proof}
    If $v\in I^*$ then $\nbri{v}=0$ which means that the expected size of \nbrit{v} is $(1/2 - \eps) \deg(v)$.
    Since we assume complete independence for the oracle we can use the Chernoff bound to get concentration.

    Let $X_i=1$ if $i^{th}$ neighbor is claimed to be in $I^*$ by the oracle where $i \in [\deg(v)]$.
    Observe that $\nbrit{v}=\sum_i X_i$ is the number of neighbors that claim to be in $I^*$.
    We know $\mu=\expect{\nbrit{v}} = (1/2 - \eps) \deg(v)$.
    Using the Chernoff (\Cref{prop:chernoff-mul}) bound with $\delta_v=6\paren{\frac{\ln n}{\deg(v)}}^{0.5} \leq 1$:
    \begin{align*}
        \prob{\nbrit{v} > (1+\delta_v)\mu} &\leq \exp\paren{-\frac{\delta_v^2 \cdot \mu}{3}} \leq n^{-3} \tag{since $\eps \leq 1/4$}.
    \end{align*}
    Notice that as $\deg(v)$ gets larger we get better concentration.
\end{proof}

Note that \Cref{clm:vert-in-I} does \emph{not} rule out the case that a vertex $v\in V\setminus I^*$ and has very few neighbors in $I^*$. Nevertheless, it tells us that if we simply eliminate the vertices that ``block'' a large number of neighbors in $I^*$, we will not mistakenly drop vertices in $I^*$.

Next, we show that if a vertex $v$ has many neighbors in $I^*$ i.e. \nbri{v} is large then \nbrit{v} should also be large and hence we should be able to detect such a vertex $v \not\in I^*$.
\begin{claim}\label{clm:vert-not-in-I}
    If $v\not\in I^*$ and $\nbri{v}\geq (3/\eps) \sqrt{\ln n} \sqrt{\deg(v)}$ then with high probability, $\nbrit{v} > (1/2-\eps)\deg(v) + 6\sqrt{\ln n} \cdot (1/2-\eps) \sqrt{\deg(v)}$.
\end{claim}
\begin{proof}
    If $v\not\in I^*$ and $\nbri{v}= k$ then the expected size of \nbrit{v} is 
    \[
        \mu=\expect{\nbrit{v}} = k(1/2+\eps) + (\deg(v)-k) (1/2-\eps) = (1/2-\eps) \deg(v) + 2\eps k.
    \]
    We now use the Chernoff bound (\Cref{prop:chernoff}) with $t=\eps k$ to get concentration:
    \begin{align*}
        \prob{\nbrit{v} < \mu - t} &\leq \exp\paren{-2t^2/\deg(v)} \\
        &= \exp\paren{-2\eps^2 k^2/\deg(v)} \\
        &\leq n^{-3}. \tag{using the lower bound on k}     
    \end{align*}
    Thus, with high probability we have:
    \begin{align*}
    \nbrit{v} &\geq \mu - \eps k \\
    &= (1/2-\eps)\deg(v) +\eps k \\
    &= (1/2-\eps)\deg(v) +3\sqrt{\ln n} \sqrt{\deg(v)} \\
    &> (1/2-\eps)\deg(v) + 6\sqrt{\ln n} \cdot (1/2-\eps) \sqrt{\deg(v)} \, .
    \end{align*}
\end{proof}

We can conclude that the events in \Cref{clm:vert-in-I} and \Cref{clm:vert-not-in-I} happen with high probability by a union bound over all vertices.

% We now prove \Cref{thm:persist-noise}.
\begin{proof}[\textbf{Finalizing the proof of \Cref{thm:persist-noise}}]
    Calculating $\nbrit{v}$ for all vertices $v\in V$ and finding set $S$ takes $O(m)$ time.
    The greedy MIS can also be computed in $O(m)$ time.

    We first condition on the events in \Cref{clm:vert-in-I} and \Cref{clm:vert-not-in-I} for all vertices (this happens with high probability).
    Notice that for all vertices in $v \in S$ we have $\nbrit{v} \leq s_v$.
    By \Cref{clm:vert-in-I} all vertices in $I^*$ are in $S$.
    By \Cref{clm:vert-not-in-I} we know that any non-MIS vertices $v$ that are in $S$ have $\nbri{v}\leq (3/\eps) \sqrt{\ln n} \sqrt{\deg(v)} \leq (6/\eps) \sqrt{\Delta \ln n}$.
    Also, vertices in $L$ have $\nbri{v}\leq \deg(v) =  \sqrt{\deg(v)} \cdot \sqrt{\deg(v)} \leq  \sqrt{\Delta} \sqrt{36\ln n} \leq (6/\eps) \sqrt{\Delta \ln n}$.

    This means that when we run the greedy MIS algorithm and pick a non-MIS vertex, we eliminate at most $(6/\eps) \sqrt{\Delta \ln n}$ vertices in $I^*$.
    Thus, we have $\card{I}\geq \frac{\eps}{6} \cdot (\Delta \ln n)^{-0.5} \cdot \card{I^*}$.
    Finally, because of the assumption on $\eps$ ($\eps \leq 1/4$), we lose a factor of at most $2$ in the approximation, giving us the final bound $\card{I}\geq \frac{\eps}{12} \cdot (\Delta \ln n)^{-0.5} \cdot \card{I^*}$.
\end{proof}

\begin{remark}
    We assume that we have complete independence between the oracle queries for the vertices. But we can get essentially the same result (up to constants) when the oracle answers the queries using a $k$-wise independent hash function instead of a completely random function for $k=O(\log n)$.

    This holds because we use \Cref{prop:chernoff-limited} with $k=O(\log n)$ instead of the Chernoff bound (\Cref{prop:chernoff-mul}). The min in the exponent always picks the second term because $k$ is large enough and so we get something very similar to the Chernoff bound in \Cref{prop:chernoff-mul} where the exponent only differs by some constants. Thus, the approximation we get will be a small constant factor worse but will remain the same asymptotically.
\end{remark}

%\prat{1. Pairwise independence 2. Barrier in going below $\sqrt{\Delta}$}

\section{An Algorithm in the Non-persistent Noise Setting}
\label{sec:alg-mabs}
In this section, we consider algorithms in the \emph{non-persistent noise setting (MABs setting)} of the MIS oracle, i.e., the algorithm can access the learning-augmented MIS oracle with \emph{fresh randomness} for each query of a vertex $v$. The formal statement of our main result in this setting is as follows.
\begin{theorem}
\label{thm:MABs-setting}
There exists a randomized algorithm that given a parameter $\delta \in (0,1)$ and
\begin{enumerate}[label=\alph*)]
\item an input graph $G=(V,E)$ with a maximum independent set $I^*$; and 
\item an MIS oracle $\orc$ in the \emph{non-persistent noise setting},
\end{enumerate}
with probability at least $(1-\delta)$, in $O(m\log{n})$ time and $\frac{30 n}{\eps^2} \cdot \log{\frac{1}{\delta}}$ total queries to $\orc$, computes a set $I$ such that $\card{I}\geq \frac{48}{50} \cdot \card{I^*}$.
\end{theorem}

We dedicate the remainder of this section to the proof of \Cref{thm:MABs-setting}.

\paragraph{The algorithm.} As we have discussed in our high-level overview, our algorithm proceeds in two phases. In the first phase, our algorithm focuses on eliminating most of the vertices in the non-MIS vertex set. 
Then, in the second phase, we show that a good approximation to vertex cover is enough to get a good approximation to the independent set.
We can easily find a $2$-approximate vertex cover in $O(m)$ time by computing a maximal matching and picking all its endpoints.
The detailed description of the algorithm is as follows.

\begin{Algorithm}
\label{alg:mis-bandit}
%An algorithm that computes an $O(1)$-approximation of the maximum independent set.\\
An algorithm for MIS in non-persistent noise setting. \\
\textbf{Input:} A graph $G=(V,E)$ that contains an unknown maximum independent set $I^*$; an MIS oracle $\orc$ in the multi-armed bandit setting; a confidence parameter $\delta\in (0,1)$. \\
\textbf{Output}: A set of vertices $I$ such that $I$ forms an independent set and $\card{I}= O(\card{I^*})$.\\
\textbf{Parameters:} $q_{r}=\frac{4}{\eps^2} \cdot (r+\log{\frac{1}{\delta}})$.
\smallskip
\begin{itemize}
\item Maintain a set of $V_r$ with the initialization $V_0 \gets V$.
\item For $r=1$ to $\infty$, do the following:
\begin{enumerate}
\item \textbf{Elimination phase:}
\begin{itemize}
\item For each vertex $v\in V_{r-1}$:
\begin{enumerate}
\item Query $v$ for $q_r$ times.
\item \emph{Remove} $v$ from $V_{r-1}$ if the number of $1$ returned by $\orc(v)$ (``yes'' answers) is less than $q_r/2$.
\end{enumerate}
\item Let the updated vertex set be $V_r$, i.e., $V_r$ is a subset of vertices of $V_{r-1}$ that gets at least $q_r/2$ ``yes'' answers from \orc(v).
\end{itemize}
\item \textbf{Vertex Cover phase:}
\begin{enumerate}
\item Compute a $2$-approximate vertex cover $U_r$ of the induced subgraph $G[V_r]$.
\item Let $I_r\gets V_r \setminus U_r$.
\end{enumerate}
\item Maintain the set $I$ with the maximum size among all $I_r$'s, i.e., let $I\gets I_r$ if $I_r$ is larger than $I$ and keep $I$ unchanged otherwise. 
\item If the total number of $\orc$ queries is more than $30 \cdot \frac{n}{\eps^2} \cdot \log{\frac{1}{\delta}}$ then terminate and return the currently maintained $I$.
\end{enumerate}
\end{itemize}
\end{Algorithm}

Note that since we do \emph{not} necessarily know the actual size of $I^*$, we compute a vertex cover after every elimination phase and simply output the independent set with the largest size throughout the process.

% \chen{The analysis is based on $p=2/3$ now -- come back and change it to $1/2+\eps$ in general. The constant will also look much better.}
\paragraph{The analysis.} We now proceed to the analysis of the algorithm. 
Before diving into the main lemmas, we first show some straightforward technical claims that characterize the behavior of the MIS and non-MIS vertices in the elimination phase. We first show that the probabilities of an MIS vertex being eliminated and a non-MIS vertex surviving in round $r$ are both small.
\begin{claim}
\label{clm:mabs-elimination-bounds}
The following statements are true:
\begin{enumerate}
\item Let $v\in V_{r-1} \cap I^*$; then, the probability that $v$ is removed from $V_{r}$ is at most $\frac{1}{100}\cdot \frac{\delta}{4^r}$.
\item Let $v\in V_{r-1} \setminus I^*$; then, the probability that $v$ is \emph{not} removed from $V_{r}$ is at most $\frac{1}{100}\cdot \frac{\delta}{4^r}$.
\end{enumerate}
\end{claim}
\begin{proof}
We prove this claim by applying the Chernoff bound in \Cref{prop:chernoff}. 
For any vertex $v\in I^*$, let the random variable $X_v^i=1$ if the $i^{th}$ query for vertex $v$ is a ``yes'' and $X_v^i=0$ otherwise for $i\in [q_r]$. 
Observe that $X_v = \sum_i X_v^i$ is the number of ``yes'' answers returned by $\orc(v)$ out of the $q_r$ queries. 
Clearly, we have that $\expect{X_v}= (1/2+\eps) \cdot q_r$, and $X_v$ is a summation of the independent indicator random variables so, we can apply \Cref{prop:chernoff} to show that
\begin{align*}
\Pr\paren{X_v < \frac{q_r}{2}}& = \Pr\paren{X_v - \expect{X_v} \leq -\eps \cdot q_r}\\
&\leq \exp\paren{-2 \cdot\eps^2 \cdot q_r} \tag{applying \Cref{prop:chernoff}}\\
&= \exp\paren{- 8r - 8\log{\frac{1}{\delta} }} \tag{by the definition of $q_r$}\\
& \leq \exp\paren{-6} \cdot \exp\paren{- 2 r} \cdot \exp \paren{- 8\log{\frac{1}{\delta}}} \\ &\leq \frac{1}{100} \cdot \frac{\delta}{4^r}.
\end{align*}
Note that the vertices in $v\in I^* \cap V_{r-1}$ are in $I^*$. Therefore, we can get the desired statement for $v \in I^* \cap  V_{r-1}$.

We can similarly define $Y_v$ for the number of ``yes'' answers returned by $\orc(v)$ with $q_r$ queries for a vertex $v \in V \setminus I^*$. Here, we have that $\expect{Y_v}= (1/2-\eps) q_r$. As such, we have that 
\begin{align*}
\Pr\paren{Y_v \geq \frac{q_r}{2}}& = \Pr\paren{Y_v - \expect{Y_v} \geq \eps \cdot q_r}\\
&\leq \exp\paren{-2 \cdot\eps^2 \cdot q_r} \tag{applying \Cref{prop:chernoff}}\\
&= \exp\paren{- 8r - 8\log{\frac{1}{\delta} }} \tag{by the definition of $q_r$}\\
& \leq \exp\paren{-6} \cdot \exp\paren{- 2 r} \cdot \exp \paren{- 8\log{\frac{1}{\delta}}} \\ &\leq \frac{1}{100} \cdot \frac{\delta}{4^r}.
\end{align*}
This gives us the desired statement for $v\in V_{r-1}\setminus I^*$ as well. \myqed{\Cref{clm:mabs-elimination-bounds}}
\end{proof}

We now prove the main technical lemma of our algorithm that helps eventually prove \Cref{thm:MABs-setting}.
In what follows, we will denote the size of $I^*$ as $\alpha n$ for some $\alpha\in (0,1)$. Our main lemma for the elimination phase is as follows.
\begin{lemma}
\label{lem:vertex-elim}
Let $\card{I^*} = \alpha n$ for some $\alpha \in (0,1)$ and $\tilde{r}=1 + \log{\frac{1}{\alpha}}$. With probability at least $1-\delta$, the following statements about \Cref{alg:mis-bandit} are true:
\begin{enumerate}[label=\roman*).]
\item The number of vertices in $V_{\tilde{r}}$ that are \emph{not} in $I^*$ is at most $\nicefrac{\alpha n}{100}$, i,e, 
\[\card{V_{\tilde{r}} \setminus I^*}\leq \frac{\alpha n}{100}.\]
\item The number of vertices in $V_{\tilde{r}}$ that are in $I^*$ is at least $\nicefrac{49}{50} \cdot \alpha n$, i.e.,
\[\card{V_{\tilde{r}} \cap I^*}\geq \frac{49}{50}\cdot \alpha n.\]
\item The \emph{total} number of $\orc$ queries in the first $\tilde{r}$ rounds is at most $\nicefrac{30 n}{\eps^2} \cdot \log{1/\delta}$, i.e., 
\[\sum_{r=1}^{\tilde{r}} \card{V_{r-1}} \cdot q_r \leq 30\cdot \frac{n}{\eps^2} \cdot \log{\frac{1}{\delta}}.\]
\end{enumerate}
Note that in the above, $ \card{V_{r-1}} \cdot q_r$ is exactly the number of queries used in round $r$.
\end{lemma}
\begin{proof}
We prove the statements in order.
\paragraph{Proof of $i)$.} Note that by \Cref{clm:mabs-elimination-bounds}, the probability that a vertex in $V \setminus I^*$ survives round $r$ is at most $\frac{1}{100}\cdot \frac{\delta}{4^r}$. As such, we have that
\begin{align*}
\expect{\card{V_{\tilde{r}} \setminus I^*}} & = \sum_{v \in V_{\tilde{r}-1}\setminus I^*} \Pr\paren{\text{$v$ survives round $\tilde{r}$}} \\
& = \sum_{v \in V \setminus I^*} \Pr\paren{\text{$v$ survives all rounds till $\tilde{r}$}}\\
& = \sum_{v \in V \setminus I^*} \prod_{i=1}^{\tilde{r}} \Pr\paren{\text{$v$ survives round $i \mid$  $v$ survives all rounds till $i-1$}} \tag{All rounds are independent}\\
& \leq \sum_{v \in V \setminus I^*} \prod_{i=1}^{\tilde{r}} \frac{\delta}{100} \cdot \frac{1}{4^i} \\
& \leq n \cdot \paren{\frac{\delta}{100}}^{\tilde{r}} \cdot \paren{\frac{1}{4}}^{\binom{\tilde{r}}{2}} \\
& \leq \frac{\delta n}{100} \cdot \paren{\frac{1}{4}}^{\tilde{r}} \\
& \leq \frac{\alpha \cdot n \cdot \delta}{400}. \tag{using $\alpha\in (0,1)$}
\end{align*}
Therefore, by Markov inequality, we have
\begin{align*}
\Pr\paren{\card{V_r \setminus I^*} > \frac{\alpha n}{100}} 
\leq \frac{\delta}{4}
\end{align*}
as desired.
\paragraph{Proof of $ii)$.} By \Cref{clm:mabs-elimination-bounds}, the probability that a vertex $v$ is eliminated in round $r$ is at most $\frac{\delta}{100}\cdot \frac{1}{4^r}$. 
We analyze the number of vertices in $I^*$ that are eliminated by round $r$. We can show that the expected value is 
\begin{align*}
\expect{\card{I^* \setminus V_{\tilde{r}}}} & = \sum_{v \in I^*} \Pr\paren{\text{$v$ is eliminated by round ${\tilde{r}}$}}\\
&\leq \sum_{v \in I^*} \sum_{i=1}^{\tilde{r}} \Pr\paren{\text{$v$ is eliminated in round $i$}} \tag{Union Bound}\\
&\leq \sum_{v \in I^*} \sum_{i=1}^{\tilde{r}} \frac{\delta}{100} \cdot \frac{1}{4^i}\\
&\leq (\alpha n) \cdot \frac{\delta}{100} \cdot \frac{1}{3} \tag{Geometric Sum}.
\end{align*}
Therefore, by a simple Markov bound, we have that
\begin{align*}
\Pr\paren{\card{I^* \setminus V_r} > \frac{\alpha n}{50}} &\leq \frac{\delta}{6}.
\end{align*}
Thus, with probability at least $1-\delta/6$ we have $\card{I^* \cap V_{\tilde{r}}} \geq \frac{49}{50} \cdot \alpha n$.

\paragraph{Proof of $iii)$.} 
Note that we are proving this bound holds even if we remove the termination condition from the algorithm. This will show that we will reach round $\tilde{r}$ with high probability.
We first condition on the events in the proofs of $i)$ and $ii)$. Note that, unlike the standard analysis of elimination-based algorithms, here, we cannot directly upper-bound the total number of queries each round. Instead, we separately analyze the number of queries induced by the vertices in $I^*$ and $V\setminus I^*$.

We first analyze the number of queries induced by the vertices in $V\setminus I^*$. Let us define $X_{\neg I^*}$ as the total number of queries induced by the non-MIS vertices. Similarly, we can define $X^{r}_{\neg I^*}$ as the queries induced by the non-MIS vertices at round $r$. Thus, we have that 
\begin{align*}
    \expect{X_{\neg I^*}} &= \sum_{v \in V-I^*} \sum_{i=1}^{\tilde{r}} \prob{v \text{ survives till round } i} \cdot q_i\\ 
    &\leq n \sum_{i=1}^{\tilde{r}} q_i \prod_{j=1}^{i} \prob{v \text{ survives round } j \mid v \text{ survives till round } j-1 } \\
    &\leq n \sum_{i=1}^{\tilde{r}} q_i \prod_{j=1}^{i} \frac{\delta}{100} \cdot \frac{1}{4^j} \tag{\Cref{clm:mabs-elimination-bounds}} \\
    &\leq n \sum_{i=1}^{\tilde{r}} \paren{\frac{\delta}{100}}^i \cdot \paren{\frac{1}{4}}^{\binom{i}{2}} \cdot q_i \\   
    &= n \sum_{i=1}^{\tilde{r}} \paren{\frac{\delta}{100}}^i \cdot \paren{\frac{1}{4}}^{\binom{i}{2}} \cdot \frac{4}{\eps^2} \cdot 
    \paren{i+ \log 1/\delta}  \\    
    &\leq \frac{4 \delta n}{100 \eps^2} \sum_{i=1}^{\tilde{r}} \paren{\frac{1}{4}}^{i}\cdot \paren{i+ \log 1/\delta} \tag{Since $\delta\leq 1$} \\    
    &\leq \frac{4 \delta n}{100 \eps^2} \paren{1 + \log 1/\delta}. \tag{using properties of geometric sums} \\    
\end{align*}
Therefore, by Markov inequality, we can show that
\begin{align*}
    \prob{X_{\neg I^*} > \frac{2n}{5\eps^2} \log 1/\delta} \leq \delta/5.
\end{align*}

We now analyze the queries induced by the vertices in $I^*$. Similar to the case of the non-MIS analysis, let us define $X_{I^*}$ as the total number of queries induced by the MIS vertices. We will trivially upper bound $X_{I^*}$ in the following way:
\begin{align*}
    X_{I^*} &\leq \alpha n \sum_{i=1}^{\tilde{r}} q_i \\
    &= \alpha n \sum_{i=1}^{\tilde{r}} \frac{4}{\eps^2} \cdot \paren{i+\log{\frac{1}{\delta}}} \\
    &\leq \frac{4 \alpha n}{\eps^2} \cdot \paren{\tilde{r}^2+ \tilde{r} \cdot \log{\frac{1}{\delta}}} \\
    &\leq \frac{4 \alpha n}{\eps^2} \cdot \paren{1 + \paren{\log 1/\alpha}^2+ \lg 1/\alpha \cdot (2+\log 1/\delta) + \log 1/\delta} \\
    &\leq \frac{4 n}{\eps^2} \cdot \paren{5+2\log 1/\delta} \tag{using $\alpha \cdot \lg{\frac{1}{\alpha}}\leq 1$ and $\alpha \cdot \lg^2{\frac{1}{\alpha}}\leq 2$ for any $\alpha \in (0,1)$}
\end{align*}

% \begin{align*}
% X_{I^*} &= \sum_{r=1}^{\tilde{r}} X^{r}_{I^*}\\
% &\leq \sum_{r=1}^{\tilde{r}} \sum_{v \in I^*} \paren{\text{number of queries used by $\orc(v)$}} \\
% &\leq \sum_{r=1}^{\tilde{r}} \alpha \cdot \frac{n}{\eps^2} \cdot (10 + 5r + 5\log{\frac{1}{\delta}})\\
% &\leq \alpha \cdot \frac{n}{\eps^2} \cdot (10 + 50\log{\frac{1}{\alpha}} + 5\log{\frac{1}{\delta}}) \tag{$\tilde{r}\leq 10 \log{\frac{1}{\alpha}}$}\\
% &\leq 10 \alpha \cdot \frac{n}{\eps^2} + 50 \alpha \cdot \frac{n}{\eps^2} \cdot \log{\frac{1}{\alpha}} + 5 \alpha \cdot \frac{n}{\eps^2} \cdot \log{\frac{1}{\delta}}\\
% &\leq 50 \cdot\frac{n}{\eps^2} \cdot \log{\frac{1}{\delta}}. \tag{using $\alpha \cdot \log{\frac{1}{\alpha}}\leq 1$ for any $\alpha >0$}
% \end{align*}
We can then add the number of queries used by $X_{\neg I^*}$ and $X_{I^*}$ to get the desired sample complexity bound of $\frac{30 n}{\eps^2} \cdot \log 1/\delta$.

Finally, we can apply a union bound over the failure probabilities of the events in the proofs of $i)$, $ii)$, and $ii)$ to argue that with probability at least $1-\delta$, all the statements hold. \myqed{\Cref{lem:vertex-elim}}
\end{proof}

We now proceed to show the guarantee of the matching and MIS phase. Our main lemma for this part is as follows.
\begin{lemma}
\label{lem:matching-mis}
Let $V_r\subseteq V$ be any subset of vertices in \Cref{alg:mis-bandit}. Furthermore, assume that the number of MIS vertices in $V_r$ is at least $50$ times the number of non-MIS vertices in $V_r$, i.e.,
\begin{align*}
\card{V_r\cap I^*} \geq 50 \cdot \card{V_r \setminus I^*}.
\end{align*}
Then, the set $I_r$ returned by \Cref{alg:mis-bandit} is a valid independent set, and we have
\[\card{I_r}\geq \frac{49}{50}\cdot \card{V_r\cap I^*}.\]
\end{lemma}
\begin{proof}
Recall that we compute a 2-approximate vertex cover $U_r$ in the vertex cover phase. We know that the complement $I_r \gets V_r \setminus U_r$ is an independent set. 
This is because all edges of the graph are incident on the vertex cover so the remaining vertices form an independent set.

We know that $V_r \setminus I^*$ is a vertex cover since $V_r\cap I^*$ is an independent set.
Thus, we have 
\begin{align*}
    \card{I_r} &= \card{V_r} - \card{U_r} \tag{by definition} \\
    &\geq \card{V_r \cap I^*} + \card{V_r \setminus I^*} - 2\card{V_r \setminus I^*} \tag{since $U_r$ is a $2$-approximation} \\
    &\geq \card{V_r \cap I^*} - \frac{1}{50} \cdot \card{V_r \cap I^*} \tag{using the assumption} \\
    &= \frac{49}{50} \cdot \card{V_r \cap I^*} 
\end{align*}
\myqed{\Cref{lem:matching-mis}}
% We then argue that the size of $I_r$ is at least $\frac{49}{50}\cdot \card{V_r\cap I^*}$. To see this, note that since there is \emph{no edges} between the vertices in $V_r\cap I^*$, for any vertex $v\in V_r\cap I^*$ to participant in the maximal matching, there must be a corresponding (distinct) vertex in $V_r\cap (V\setminus I^*)$. Furthermore, since there are at most $\frac{1}{50}\cdot \card{V_r\cap I^*}$ vertices in $V_r\cap (V\setminus I^*)$, we will remove at most $\frac{1}{50}\cdot \card{V_r\cap I^*}$ vertices from $V_r\cap I^*$, which results in the desired bound. 
\end{proof}

The final missing piece is the \emph{efficiency} of the algorithm. We now prove that the algorithm is efficient both in time and the number of $\orc$ oracle queries.
\begin{lemma}
\label{lem:bandit-alg-efficiency}
\Cref{alg:mis-bandit} runs in $O(m \log{n})$ time and uses at most $\frac{30 n}{\eps^2} \cdot \log{\frac{1}{\delta}}$ queries on $\orc$.
\end{lemma}
\begin{proof}
The query complexity is by the design of the algorithm as we terminate upon using more than $30 \cdot \frac{n}{\eps^2} \cdot \log{\frac{1}{\delta}}$ queries. 

For the running time, note that in each iteration of $r$, we only need to: $i).$ take the majority for all queried vertices, which can be maintained in $O(n)$ time; and $ii).$ compute a greedy matching and remove the vertices, which takes $O(m)$ time. By \Cref{lem:vertex-elim}, the process terminates in $O(\log{\frac{1}{\alpha}}) = O(\log{n})$ time ($\alpha \geq \frac{1}{n}$ since there has to be at least one vertex in $I^*$). Therefore, the entire algorithm takes $O(m\log{n})$ time in total.
\end{proof}

\begin{proof}[\textbf{Finalizing the proof of \Cref{thm:MABs-setting}}] 
The query efficiency is by the design of the algorithm, and the running time simply follows from \Cref{lem:bandit-alg-efficiency}. For the approximation guarantee, note that by \Cref{lem:vertex-elim}, we will proceed to round $\tilde{r}=10 \log{\frac{1}{\alpha}}$, at which point we will have $\card{V_{\tilde{r}}\cap I^*}\geq \frac{49}{50}\cdot \alpha n$ and $\card{V_r\cap I^*} \geq 50 \cdot \card{V_r \setminus I^*}$. Therefore, by \Cref{lem:matching-mis}, the returned $I_{\tilde{r}}$ is of size at least
\begin{align*}
\card{I_{\tilde{r}}}\geq \frac{49}{50}\cdot \card{V_{\tilde{r}}\cap I^*} \geq \frac{49}{50}\cdot \frac{49}{50} \cdot \alpha n,
\end{align*}
which gives us the desired $48/50$ approximation.
\end{proof}
% \vihan{We should change the constants to get something like a $9/10$-apx or maybe $(1-\gamma)$-apx}

\begin{remark}
We aim to get the $O(1)$ approximation in our algorithm and analysis. However, we remark that we can get both non-asymptotic and asymptotic trade-offs between the number of queries and the approximation factor. For the non-asymptotic trade-off (i.e., using more queries to get a better constant approximation), we can increase the leading constant on the sample complexity, and obtain the approximation with a larger constant. For the asymptotic trade-off, we can perform the simple trick by sampling $k$ vertices uniformly at random and running \Cref{alg:mis-bandit} on the sampled vertices. This will give us an $O(\frac{k}{n})$-approximation algorithm with $O(\frac{k}{\eps^2}\cdot \log{\frac{1}{\delta}})$ queries as long as $\alpha k = \Omega(\log{n})$. % \chen{This is not entirely true -- subsampling could give us $\alpha k<\log{n}$ and we won't have concentration. Unless we want to add those auxiliary algorithms to the appendix and write an appendix just for this,  we should drop this point.}
\end{remark}

%A small note that in this setting we do not need the knowledge of $\eps$.

\section{Discussion and Open Problems}
\label{sec:discussion}
We discussed learning-augmented algorithms for the Maximum Independent Set problem in this paper. Our main results include algorithms for both persistent and non-persistent noise settings, demonstrating that a learning-augmented oracle could lead to MIS algorithms with considerably better efficiency. There are several intriguing open problems following our work.
\begin{itemize}
\item For the \textbf{persistent noise} setting, the main open question is whether we could beat the $\widetilde{\Theta}(\sqrt{\Delta}/\eps)$ approximation bound with the same oracle. We do not have any lower bounds for the persistent noise setting in this paper, and it is unclear what type of techniques could be used to prove lower bounds for learning-augmented algorithms.
\item For the \textbf{non-persistent noise} setting, our algorithm matches the \emph{asymptotically} optimal approximation factor using $O(n)$ queries. In \Cref{app:lb-mabs}, we also proved that we cannot obtain the same results by only querying the oracle (and \emph{not} looking into the graph). An open problem here is that if we want to recover a $1-o(1)$ fraction of the MIS vertices, how many queries do we need? We suspect there is a lower bound on the number of queries (e.g., $\omega(n)$), but it is not immediately clear how to prove it.
\item We can also ask about \textbf{sublinear} number of queries on the oracle $\orc$, i.e., if we make $o(n)$ queries on the oracle, what is the best we can do for both persistent and non-persistent noise settings? Currently, our algorithms in both settings require $\Omega(n)$ queries to the oracle.
\item Finally, for the \textbf{practical} aspect of the algorithms, we believe the oracles are possible to implement in practice. For instance, if we have features on the nodes, it is possible to use forward-pass graph convolution networks (GCNs), and simply run greedy in each ``cluster'' of nodes whose final features are sufficiently similar. Exploring practical oracles for this purpose would also be an interesting problem to resolve.
\end{itemize}

\subsection*{Acknowledgements}
The authors are grateful to Sepehr Assadi and Samson Zhou for the helpful conversations 
regarding the project. The authors also thank anonymous APPROX reviewers for helpful suggestions.

\bibliographystyle{alpha}
\bibliography{reference}

\clearpage

\appendix

\section{Technical Preliminaries}
\label{app:tech-prelim}

We use the following standard forms of Chernoff bound. % \vihan{put this in a separate file}

\begin{proposition}[Chernoff-Hoeffding bound]\label{prop:chernoff}
Let $X_1,\ldots,X_m$ be $m$ independent random variables with support in $[0,1]$. Define $X := \sum_{i=1}^{m} X_i$. Then, for every $t > 0$, 
\begin{align*}
    & \Pr\paren{X - \expect{X} \geq t} \leq  \exp\paren{-\frac{2t^2}{m}}\\
    & \Pr\paren{X - \expect{X} \leq -t} \leq  \exp\paren{-\frac{2t^2}{m}}. 
\end{align*}
\end{proposition}

\begin{proposition}[Chernoff bound; 
c.f.~\cite{dubhashi2009concentration}]\label{prop:chernoff-mul}
	Suppose $X_1,\ldots,X_m$ are $m$ independent random variables with range $[0,1]$ each. Let 
	$X := \sum_{i=1}^m X_i$ and $\mu_L \leq \expect{X} \leq \mu_H$. Then, for any $\delta \in [0,1]$, 
	\[
	\Pr\paren{X >  (1+\delta) \cdot \mu_H} \leq \exp\paren{-\frac{\delta^2 \cdot 
	\mu_H}{3+\delta}} \quad 
	\textnormal{and} \quad \Pr\paren{X <  (1-\delta) \cdot \mu_L} \leq \exp\paren{-\frac{\delta^2 
	\cdot 
			\mu_L}{2+\delta}}.
	\]
\end{proposition}
 
 We also consider limited independence hash functions.
 Roughly speaking, a $k$-wise independent hash function behaves like a totally random function when considering at most $k$ elements.
 Formally, a family of hash functions $H= \set{h:[n] \rightarrow [m]}$ is $k$-wise independent if for any $x_1,x_2,\ldots,x_k \in [n]$ and $y_1,y_2,\ldots,y_k \in [m]$ the following holds:
 \[
    \Pr_{h \in_R H}\paren{h(x_1)=y_1 \wedge h(x_2)=y_2 \wedge \ldots \wedge h(x_k)=y_k} = m^{-k}.
 \]
 % We use the following standard result for $k$-wise independent hash functions.
 % \begin{proposition}[\!\!\cite{MotwaniR95}]\label{prop:k-wise}
 % 	For every integers $n,m,k \geq 2$, there is a $k$-wise independent hash function $\mathcal{H} = \set{h: [n] \rightarrow [m]}$ so that sampling and storing a function $h \in \mathcal{H}$ takes $O(k \cdot (\log n + \log m))$ bits of space.
 % \end{proposition}
 We shall use the following concentration result on an extension of Chernoff-Hoeffding bounds for limited independence hash function. 
 \begin{proposition}[\cite{SchmidtSS95}]\label{prop:chernoff-limited}
 	Suppose $h$ is a $k$-wise independent hash function and $X_1,\ldots,X_m$ are $m$ random variables in $\set{0,1}$ where $X_i = 1$ iff $h(i) = 1$. 
	Let $X := \sum_{i=1}^{m} X_i$. Then, for any $\delta > 0$, 
	\[
		\Pr\paren{\card{X - \expect{X}} \geq \delta \cdot \expect{X}} \leq \exp\paren{-\min\set{\frac{k}{2},\frac{\delta^2}{4+2\delta} \cdot \expect{X}}}. 
	\] 
 \end{proposition}

\section{A Simple Algorithm for an \texorpdfstring{$O(\log{n})$}{O(log n)}-approximation in the Non-persistent Noise Setting}
\label{app:mabs-simple-alg}

In this section, we give a simple algorithm to get an $O(\log{n})$-approximation using the non-persistent noise oracle.
The idea is to sample every vertex independently with probability $1/\log n$ and the query the sampled vertices $O(\log n)$ times.

\begin{lemma}
    There is an algorithm that gives an $O(\log{n})$-approximation using the non-persistent noise oracle with $O(n/\eps^2)$ queries with constant probability.
\end{lemma}
\begin{proof}
    We sample every vertex independently with probability $1/\log n$ and add it to the set $U$.
    For every vertex $u \in U$, query it $q:=\ln n/\eps^2$ times and compute the majority of the queries.
    
    The claim now is that this is enough to know whether $u\in I^*$ for every $u\in U$ with high probability.
    We prove this claim by applying the Chernoff bound in \Cref{prop:chernoff}. 
    For any vertex $v\in U \cap I^*$, let the random variable $X_v^i=1$ if the $i^{th}$ query for vertex $v$ is a ``yes'' and $X_v^i=0$ otherwise for $i\in [q]$. The proof is the same for $v\in u \setminus I^*$.
    Observe that $X_v = \sum_i X_v^i$ is the number of ``yes'' answers returned by $\orc(v)$ out of the $q$ queries. 
    Clearly, we have that $\expect{X_v}= (1/2+\eps) \cdot q$, and $X_v$ is a summation of the independent indicator random variables so, we can apply \Cref{prop:chernoff} to show that
    \begin{align*}
    \Pr\paren{X_v < \frac{q}{2}}& = \Pr\paren{X_v - \expect{X_v} \leq -\eps \cdot q}\\
    &\leq \exp\paren{-2 \cdot\eps^2 \cdot q} \tag{applying \Cref{prop:chernoff}}\\
    &= \exp\paren{-2 \cdot \ln n}\\
    &=n^{-2}.
    \end{align*}
    A union bound over all vertices implies that we know whether $u\in I^*$ for every $u\in U$ with high probability.
    Also, using the Chernoff bound, we can say that with high probability, the number of vertices that are sampled in $U$ is at most $2n/\log n$.
    Thus, the total number of queries is at most $2n/\eps^2$.

    % We also know that $\card{I^*}\geq \log^2 n$ so when we sample every vertex with probability $1/\log n$ we sample at least $\card{I^*}/2\log n$ vertices from $I^*$ with high probability.
    % This can also be shown using a Chernoff bound.
    % So we are able to recover a $O(1/\log n)$ fraction of $I^*$ giving us an $O(\log{n})$-approximation.

    We sample every vertex with probability $1/\log n$ so we sample $\card{I^*}/\log n$ vertices from $I^*$ in expectation.
    Firstly, we note that if $\card{I^*} \leq 10 \log n$ the outputting a single vertex gives a $O(\log n)$-approximation.
    So we consider the case where $\card{I^*} \geq 10 \log n$. 
    For any vertex $v\in V$, let the random variable $Y_v=1$ if the vertex $v$ is sampled in $U$ and $Y_v=0$ otherwise.
    Let $Y_{I^*} = \sum_{v\in I^*} Y_v$ be the random variable denoting the number of elements of $I^*$ sampled in $U$. 
    Clearly, we have that $\expect{Y_v}= 1/\log n$, and $Y_{I^*}$ is a summation of the independent indicator random variables so, we can apply \Cref{prop:chernoff-mul} with $\delta=0.5$ to show that:
    \begin{align*}
        \prob{Y_{I^*} < \expect{Y_{I^*}}/2} &\leq \exp\paren{-\expect{Y_{I^*}}/ 10} \\
        &= \exp\paren{-\card{I^*}/10 \log n} \\
        &\leq \exp\paren{-10 \log n/10 \log n} \tag{Using the lower bound on $\card{I^*}$} \\
        &=1/e.
    \end{align*}
    Thus, with constant probability we sample at least $\card{I^*}/2\log n$ vertices from $I^*$.
    So we are able to recover a $O(1/\log n)$ fraction of $I^*$ giving us an $O(\log{n})$-approximation.
    \end{proof}

    Observe that if $\card{I^*} = \omega(\log n)$ then the above lemma holds with probability $1-o(1)$.
    \Cref{alg:mis-bandit} improves the approximation factor from $O(\log n)$ to $O(1)$.

\section{A Lower Bound for the Non-persistent Noise Setting}
\label{app:lb-mabs}
In this section, we give lower bounds showing that it is not possible to solve the non-persistent noise setting with $O(n/\eps^2)$ queries to $\orc$ using \emph{only} the pure exploration algorithms in MABs. This lower bound shows that to get the desired $O(1)$-approximation with $O(n/\eps^2)$ samples, we must use the fact that the oracle is based on a maximum independent set and utilize structural properties therein.

In what follows, we divide the lower bounds into two results: the first lower bound is for returning the \emph{entire} MIS set $I^*$, and the second lower bound works with the stronger setting where we only need to return an $O(1)$ fraction of the vertices in $I^*$. A caveat here is that for technical reasons, our second lower bound works with arms with \emph{Gaussian} distributions (see for details). A stronger lower bound with Bernoulli distributions is still interesting to pursue.

\subsection{A lower bound for obtaining the entire maximum independent set}
We give a lower bound for getting the entire set $I^*$ of the maximum independent set. 
The starting point of our lower bound is a top-$k$ pure exploration lower bound by \cite{KalyanakrishnanTAS12}. The statement of the lower bound is as follows.
\begin{proposition}[\cite{KalyanakrishnanTAS12}, rephrased]
\label{prop:worst-case-mab-lb}
Let $\dist(n,k)$ be a distribution of Bernoulli MABs instances as follows.
\begin{tbox}
\textnormal{
$\dist(n,k)$: a distribution of Bernoulli multi-armed bandits with $n$ arms such that $k\leq \frac{n}{2}$.
\begin{itemize}
\item In both cases, let the first arm be with mean reward $\frac{1}{2}$.
\item Sampling of the high-reward arms:
\begin{enumerate}
\item With probability $1/2$, uniformly at random sample $k$ coordinates $S\subseteq \{2,\cdots, n\}$, and let the arms with mean rewards $\frac{1}{2}+\eps$.
\item With probability $1/2$, uniformly at random sample $(k-1)$ coordinates $S\subseteq \{2,\cdots, n\}$, and let the arms with mean rewards $\frac{1}{2}+\eps$.
\end{enumerate}
\item For the rest of the arms, e.g., arms in coordinates $\{2,\cdots, n\}\setminus S$, let their mean rewards be $\frac{1}{2}-\eps$.
\end{itemize}
}
\end{tbox}
\noindent
Then, any algorithm that finds the top $k$ arms on $\dist(n,k)$ with probability at least $9/10$ has to use $\Omega(\frac{n}{\eps^2}\cdot \log{k})$ samples.
\end{proposition}

We now state our formal lower bound for obtaining the entire set $I^*$ of the maximum independent set, which is a straightforward application of \Cref{prop:worst-case-mab-lb}.
\begin{lemma}
\label{lem:mabs-lb-entire-set}
Any algorithm that given \emph{only} the oracle $\orc$ for graph with $|I^*| = \omega(1)$ (but \emph{not} the graph $G$), returns the entire set of the maximum independent set $I^*$ with probability at least $9/10$ has to use $\omega(n/\eps^2)$ queries to $\orc$.
\end{lemma}
\begin{proof}
Note that to prove the lemma statement, we only need to show that there \emph{exists} a distribution of graphs such that an algorithm that returns the entire set of $I^*$ with $q$ queries will imply an algorithm that finds the top-$k$ arm in $\dist(n,k)$ with $q$ queries. Indeed, suppose we are given an algorithm $\ALG$ that given access \emph{only} to $\orc$, returns the entire set of $I^*$ with probability at least $9/10$ and $O(n)$ queries. We can then construct a hard distribution simply as follows
\begin{tbox}
\begin{itemize}
\item With probability $1/2$, let $G=(V,E)$ be with a maximum independent set $I^*$ of size $k$.
\item With probability $1/2$, let $G=(V,E)$ be with a maximum independent set $I^*$ of size $(k-1)$.
\end{itemize}
\end{tbox}
By definition, $\ALG$ can find the entire set of $I^*$ with probability at least $99/100$ in both cases of the above distribution. 

We now construct a distribution of $\dist(n',k')$ with $n'$ and $k'$ to be specified later. In the distribution of the MIS instances, we can view the random variable induced by $\orc(v)$ for each $v\in V$ as an arm. Furthermore, we can add an extra arm with a mean reward of $1/2$. This means we have constructed a distribution of $\dist(n',k)$ with $n'=n+1$ and $k'=k$ (here, we slightly overload the notation to let $n$ be the number of vertices in $G$).

Since we know the coordinate of the special arm with mean reward $1/2$, we can simply run $\ALG$ with arms on coordinates $\{2,3,\cdots, n\}$. In the end, since we can recover the \emph{entire} set of $I^*$, we can output with the following rules:
\begin{itemize}
\item If the output of $\ALG$ is of size $k$, output the same set of arms;
\item If the output of $\ALG$ is of size $k-1$, output the set of returned arms and the special arm with mean reward $1/2$.
\end{itemize}
Note that conditioning on the success of $\ALG$, the reduction algorithm always succeeds in finding the top-$k$ arms.

Finally, let $k=n/2$, and by using \Cref{prop:worst-case-mab-lb}, to output the top-$k$ arms, we need $\Omega(\frac{n'}{\eps}\cdot \log{k'})= \Omega(\frac{n}{\eps}\cdot \log{n})$ arm pulls. As such, $\ALG$ in the parameter range of $k=n/2$ must use $ \Omega(\frac{n}{\eps}\cdot \log{n})=\omega(n)$ queries to $\orc$, as desired by \Cref{lem:mabs-lb-entire-set}.
\end{proof}

\subsection{A lower bound for obtaining an \texorpdfstring{$O(1)$}{O(1)}-approximation of the maximum independent set}
In what follows, we present a slightly more involved lower bound. The lower bound shows that if the algorithm does not look into $G$, then even the easier task to obtain a \emph{constant fraction} of the MIS set $I^*$ would require $\omega(n/\eps^2)$ queries. % \chen{State explicitly that we work with Gaussian arms for the lower bound in this section.}

For technical reasons, we work with \emph{Gaussian arms} for this lemma, for which a stronger \emph{instance-dependent} lower bound was obtained by \cite{SimchowitzJR17}. A definition of the oracle $\orc$ with Gaussian reward distributions is given as follows.
% \chen{Add a definition}
We define $\orc: V \rightarrow \mathbb{R}$ as a function that upon queried on a vertex $v \in I^*$, $\orc(v)$ returns a random sample from $\mathcal{N}(\frac{1}{2}+\eps, 1)$. Conversely, upon queried on a vertex $u \notin I^*$, $\orc(v)$ returns a random sample from $\mathcal{N}(\frac{1}{2}-\eps, 1)$.
% \prat{@Chen: Check if this okay}

For Gaussian arms, \cite{SimchowitzJR17} provides a lower bound that can be described as follows.
\begin{proposition}[\cite{SimchowitzJR17}, rephrased]
\label{prop:inst-opt-mab-lb}
Let $\dist(n,k)$ be a distribution of Gaussian multi-armed bandit instances as follows.
\begin{tbox}
\textnormal{
$\dist(n,k)$: a distribution of multi-armed bandits instances with $n$ arms.
\begin{itemize}
\item Uniformly at random sample $k$ coordinates, let the mean rewards of the corresponding arms be $\frac{1}{2}+\eps$.
\item For the rest of the $(n-k)$ arms, let the mean rewards of the corresponding arms be $\frac{1}{2}-\eps$.
\end{itemize}
}
\end{tbox}
\noindent
Then, any algorithm that returns the top $k$ arms from $\dist(n,k)$ with probability at least $9/10$ has to use at least $C\cdot \frac{k}{\eps^2}\cdot \log(n-k) + \frac{(n-k)}{\eps^2} \log{k}$ arm pulls for some absolute constant $C$.
\end{proposition}
Note that \Cref{prop:inst-opt-mab-lb} only states a very weak version of the main results of \cite{SimchowitzJR17} -- the paper actually proved a much stronger lower bound. Nevertheless, as we will see shortly, the current version is already sufficient to prove a desired lower bound, stated as follows.
\begin{lemma}
\label{lem:mabs-lb-approx-set}
Any algorithm that returns a set $\hat{I}$ with probability at least $9/10$, such that 
\begin{enumerate}[label=\alph*).]
\item $\hat{I}\subseteq I^*$;
\item $\card{\hat{I}} \geq \frac{2}{3} \cdot \card{I^*}$;
\end{enumerate}
given \emph{only} the oracle $\orc$ with \emph{Gaussian reward} distributions (but \emph{not} the graph $G$) has to use $\omega(n/\eps^2)$ queries to $\orc$ in the worst case. 
\end{lemma}

% \begin{proof}
% The plan:
% \begin{enumerate}
% \item Logically, we only need to show that there exists a parameter $k$, such that an $O(n)$-query $2k/3$ algorithm can be transformed into top-$k$ algorithm with $o(\frac{n}{\eps^2} \log(n-k))$ queries.
% \item Let $k=(1-\frac{1}{\log{n}})\cdot n$.
% \item The simulation algorithm
% \begin{enumerate}
% \item For each round, run the $2k/3$ algorithm for $O(\log\log{n})$ times. Returns the set of arms that has been selected by at least $1/2$ of the times.
% \item Eliminate the set of returned arms, and recurse.
% \item If the total number of arms reaches $O(n/\log{n})$, sample each arm $O(\log{n})$ times, compute w.h.p. the correct answers.
% \end{enumerate}
% \end{enumerate}
% \end{proof}

%To prove this lemma we give an algorithm for obtaining MIS $I^*$ using a $\frac{2}{3}$-approximation algorithm for $I^*$.

We will prove the lemma by contradiction. Assume there exists a $\frac{2}{3}$-approximation algorithm for the non-persistent noise setting that uses $O(n/\eps^2)$ queries to $\orc$ (let us call it $\ALG$). We show that we can use this algorithm to recover $I^*$.

\begin{Algorithm}
\label{alg:const-to-full}
%A learning augmented algorithm that computes an $O(\sqrt{\Delta \log{n}}/\eps)$-approximation of the maximum independent set.\\
 An algorithm for obtaining MIS $I^*$ using a $\frac{2}{3}$-approximation algorithm for $I^*$.\\
\textbf{Input:} A set $V$ of $n$ vertices; access to $\orc$; a $2/3$-approximate algorithm $\ALG$. \\
\textbf{Output}: The MIS $I^*$. \\
\textbf{Parameters:} $\card{I^*}=k=\paren{1-\frac{1}{\log n}} \cdot n \, , \qquad r:=\log_{1.5} \log n \, , \qquad t:=100 \ln\ln{n}$. 
\begin{enumerate}
    \item Let $I = \emptyset$ and $V_0=V$.
    \item For $i=1$ to $r$:
    \begin{itemize}
        \item Run $\ALG$ for $t$ times. 
        \item Let $I_i$ be the set of vertices that have been selected at least $t/2$ times.
        \item Let $V_i \leftarrow V_{i-1} \setminus I_i$ and $I \leftarrow I \cup I_i$.
    \end{itemize}
    \item For each $v \in V_r$, query oracle $2\log{n}/\eps^2$ times and compute the majority answer for them.
    \item Add to $I$ all vertices whose majority answer is ``yes'' and then return $I$.
\end{enumerate}
\end{Algorithm}

\begin{claim}\label{clm:const-to-full-alg-reps}
    When the $2k/3$-approximation algorithm ($\ALG$) is run for $t$ times, it is correct $8t/10$ times with probability $1-1/\ln^2 n$.
\end{claim}
\begin{proof}
    A single run of the $2k/3$-approximation algorithm is correct with probability $9/10$ and fails with probability $1/10$.
    %We run it $t$ times independently so in expectation it fails $t/10$ times and succeeds $9t/10$ times.
    Across $t$ independent runs, in expectation it will succeed $9t/10$ times.

    Let the random variable $X_i=1$ if the $i^{th}$ run of the algorithm is correct and $X_i=0$ otherwise for $i\in [t]$.
    Observe that $X = \sum_i X_i$ is the number of times the algorithm is correct.
    We have that $\expect{X}= 9t/10$, and $X$ is a summation of the independent indicator random variables so, we can apply \Cref{prop:chernoff} to show that
    \begin{align*}
    \Pr\paren{X - \expect{X} \leq -t/10} &\leq \exp\paren{-2 t/100} \tag{applying \Cref{prop:chernoff}}\\
    &= \exp\paren{-2 \cdot \ln \ln n}\\
    &=1/\ln^2 n.
    \end{align*}
    Thus, the $2k/3$-approximation algorithm is correct $8t/10$ times with probability $1-1/\ln^2 n$.
\end{proof}

We now condition on the event in \Cref{clm:const-to-full-alg-reps} for the rest of the proof.
A vertex which is not in $I^*$ will be reported at most $2t/10<t/2$ times so it is never be added to $I$.
%The concern now is that different correct iterations can return a different $2/3$ fraction of $I^*$.
%So the question is how many vertices will be added to $I^*$ in one round?
Since different vertices can be returned as part of $2/3$-approximation across different runs, we wish to bound the number of vertices added to $I$ after each round.
Let $k_i$ be the remaining size of $I^*$ at the beginning of round $i$ i.e. $k_i=\card{I^* \setminus \cup_{j=1}^{i-1} I_j}$.
\begin{claim}\label{clm:const-to-full-alg-round-helper}
    At least $k_i/3$ vertices are added to $I$ in round $i$ i.e.\ $\card{I_i} \geq k_i/3$. 
\end{claim}
\begin{proof}
    We will prove this using a charging argument.
    For each of the $8t/10$ iterations, we will put a charge of $1$ on a vertex that was not output in the $2/3$ fraction of $I^*$.
    The total charge on a vertex corresponds to the number of times it was not output. If that exceeds $t/2$ then the vertex is not added to $I^*$ by \Cref{alg:const-to-full} in round $i$.

    The total amount of charge we distribute is $8t/10 \cdot k/3$ since only $k/3$ vertices are not output in each of the $8t/10$ correct rounds.
    This charge is distributed over $k$ vertices so the average charge on a vertex is $8t/30$.
    The number of vertices that could have a charge of $\geq t/2$ is at most $8k/15 <2k/3$.
    Thus, at least $k/3$ vertices have a charge of $< t/2$ and are added to $I^*$ in round $i$.
\end{proof}

We now show that after $r$ iterations the remaining size of $I^*$ is $k/\log n$.
\begin{claim}\label{clm:const-to-full-alg-rounds}
    After $r$ iterations the remaining size of the MIS $I^*$ is $k/\log n$.
\end{claim}
\begin{proof}
    In each round, we recover at least $1/3$ fraction of the remaining vertices of $I^*$ (\Cref{clm:const-to-full-alg-round-helper}).
    This means that for all $i \in [r]$ we have:
    \[  
    \card{I^* \setminus \cup_{j=1}^{i} I_j} \leq \frac{2}{3} \cdot \card{I^* \setminus \cup_{j=1}^{i-1} I_j}.
    \]
    Therefore we have 
    \[  
    \card{I^* \setminus \cup_{j=1}^{r} I_j} \leq \paren{\frac{2}{3}}^r \cdot \card{I^*} = k/\log n.
    \]
\end{proof}

Thus the remaining number of MIS vertices are $k/\log n \leq n/\log n$ and the total number of non-MIS vertices are $n/\log n$ (by the choice of $k$).
Therefore querying all the remaining vertices $2 \log n/\eps^2$ times gives us $4n/\eps^2$ additional queries.
%Also, we get to know which of the remaining vertices are in $I^*$ so we can recover $I^*$ completely.
As a result, with probability $1-1/n^2$ we identify the remaining vertices in $I^*$ and are able to recover it completely. A union bound over all the vertices implies the correctness of this with probability $1-1/n$.
\begin{claim}\label{clm:const-to-full-alg-result}
    \Cref{alg:const-to-full} uses $O\paren{(\log \log n)^2}$ calls to the $\frac{2}{3}$-approximation algorithm and $O(n/\eps^2)$ queries to $\orc$ and returns $I^*$ with probability $1-o(1)$.
\end{claim}
\begin{proof}
    It is easy to observe that the output of the algorithm is $I^*$.
    This is because we know that we do not add any non-MIS vertices to $I$ and in the last step we add all the remaining vertices of $I^*$ to $I$ implying that the returned value $I = I^*$.

    We know that the number of queries in the last step is $4n/\eps^2$.
    Also, the number of calls to the $\frac{2}{3}$-approximation algorithm is $r \cdot t=O\paren{(\log \log n)^2}$.

    Finally, we have conditioned on the event in \Cref{clm:const-to-full-alg-reps} and we run $O\paren{(\log \log n)^2}$ independent copies of it.
    Thus, by a union bound we get a failure probability of at most $O\paren{(\log \log n)^2/\log n}$.
    Finally, we union bound over the failure probability of $1/n$ from the last step giving us a failure probability of $O\paren{(\log \log n)^2/\log n}=o(1)$.
\end{proof}

We now finalize the proof of \Cref{lem:mabs-lb-approx-set}.
\begin{proof}[Proof of \Cref{lem:mabs-lb-approx-set}]
    We prove this by contradiction. 
    We assume towards a contradiction that there is an algorithm that returns a set $\hat{I}$ with probability at least $9/10$, such that $\hat{I}\subseteq I^*$ and $\card{\hat{I}} \geq \frac{2}{3} \cdot \card{I^*}$ using $O(n/\eps^2)$ queries to $\orc$.
    We call this the $\frac{2}{3}$-approximation algorithm.
    
    We fix $\card{I^*}=k=(1-\frac{1}{\log{n}})\cdot n$.
    We will use this algorithm to recover $I^*$ completely using $o(k/\eps^2 \log(n-k))= o(n \log n/\eps^2)$ queries.
    
    We know that \Cref{alg:const-to-full} uses $O\paren{(\log \log n)^2}$ calls to the $\frac{2}{3}$-approximation algorithm and $O(n/\eps^2)$ queries to $\orc$ and returns $I^*$.
    Also, by assumption we know that the $\frac{2}{3}$-approximation algorithm uses $O(n/\eps^2)$ queries to $\orc$.
    Thus, the total number of queries used is $O(n/\eps^2) + O(n/\eps^2) \cdot O\paren{(\log \log n)^2} = O\paren{n(\log \log n)^2/\eps^2} = o(n \log n/\eps^2)$.
    Therefore, we get a contradiction and the lemma holds.
\end{proof}

\begin{remark}
We note that \Cref{alg:mis-bandit} also works for the setting when $\orc$ returns samples drawn from Gaussian distributions. This is similar to the MAB setting with Gaussian reward distribution. This is possible because we can use Gaussian tail bound in \Cref{clm:mabs-elimination-bounds} to get similar concentration results, and the rest of the proof follows.
\end{remark}

\end{document}